\documentclass[10pt,a4paper]{article}

\usepackage[utf8]{inputenc}
\usepackage[T1]{fontenc}
\usepackage{amsmath,amssymb,amsthm}
\usepackage{booktabs}
\usepackage{geometry}
\usepackage{enumitem}
\usepackage{hyperref}
\usepackage[capitalize,noabbrev]{cleveref} 

\theoremstyle{plain}
\newtheorem{theorem}{Theorem}
\newtheorem{lemma}[theorem]{Lemma}

\theoremstyle{definition}
\newtheorem{definition}[theorem]{Definition}
\newtheorem{assumption}[theorem]{Assumption}

\theoremstyle{remark}
\newtheorem{remark}[theorem]{Remark}

\crefname{theorem}{Theorem}{Theorems}
\crefname{lemma}{Lemma}{Lemmas}
\crefname{corollary}{Corollary}{Corollaries}
\crefname{definition}{Definition}{Definitions}
\crefname{assumption}{Assumption}{Assumptions}
\crefname{remark}{Remark}{Remarks}
\crefname{section}{Section}{Sections}
\crefname{subsection}{Section}{Sections}
\Crefname{theorem}{Theorem}{Theorems}
\Crefname{lemma}{Lemma}{Lemmas}
\Crefname{corollary}{Corollary}{Corollaries}
\Crefname{definition}{Definition}{Definitions}
\Crefname{assumption}{Assumption}{Assumptions}
\Crefname{remark}{Remark}{Remarks}
\Crefname{section}{Section}{Sections}
\Crefname{subsection}{Section}{Sections}

\newcommand{\set}[1]{\mathcal{#1}}
\newcommand{\addr}{a}
\newcommand{\task}{T}
\newcommand{\stall}{\Delta t}
\newcommand{\interf}{I_{\task}}

\newcommand{\memlat}{L_{\text{mem}}}
\newcommand{\mapping}{S(\cdot)}
\newcommand{\setidx}{\sigma}
\newcommand{\tagf}{\tau}

\title{A Per-Access Upper Bound for Shared-Resource Interference \\ in Direct-Mapped Multicore Architectures}
\author{Felipe T. Pedroni\\ \textit{Founder, ExtenSilica} \\ \texttt{felipe.pedroni@extensilica.com}}
\date{\today}

\begin{document}

\maketitle

\begin{abstract}
We present a formal bounding analysis for maximum credible interference in multicore processors under strict architectural invariants: direct-mapped L2 cache (1-way associativity), disabled Miss Status Handling Registers (MSHRs), single-bank main memory, deterministic pinned tasks with fixed physical memory mapping, and a pessimistic L2/memory arbitration policy. We prove that, under these invariants, the per-critical-access stall imposed on a target task $\task$ is bounded above by $(N-1)\memlat$, and that this bound is attained by a synchronized adversarial workload of $N-1$ congruent-different-tag memory requests issued in phase with $\task$'s critical accesses. The argument is per-access and direct, requiring no informal multiplicative interference function. The derivation is purely analytical and discussed in the context of DO-178C/CAST-32A certification objectives for airborne software. Limitations and conditions for applicability are explicitly stated. This work provides a traceable method for separating multicore interference from Worst-Case Execution Time (WCET) budgets under fixed architectural constraints.
\end{abstract}


\section{Introduction}
\label{sec:intro}

Multicore processors introduce shared-resource contention channels that complicate timing analysis for safety-critical systems. Certification standards such as DO-178C with CAST-32A guidance require demonstration of ``maximum credible interference'' for shared resources including caches, buses, and memory controllers~\cite{cast32a}. However, existing approaches often rely on empirical measurement or conservative pessimism without formal closure of the adversarial search space.

This work addresses the following question: \emph{Given a fixed architectural configuration and arbitration policy, can we formally bound the worst-case shared-resource interference on a target task and exhibit an adversarial workload that attains the bound?}

We answer affirmatively under a well-defined constraint space. Our contribution is threefold:
\begin{enumerate}[label=(\roman*)]
    \item A formal system model capturing hardware invariants, workload determinism, an arbitration policy, and configuration parameters (\cref{sec:formalization});
    \item Three lemmas establishing per-dimension upper bounds on spatial, temporal, and pattern-based interference (\cref{sec:lemmas});
    \item A main theorem proving a per-access stall bound of $(N-1)\memlat$, attained by a synchronized congruent-different-tag adversarial workload (\cref{sec:theorem}).
\end{enumerate}

The result may support certification evidence packages requiring traceable, analytically justified interference bounds, provided the target platform is independently characterized to satisfy the stated invariants. It also aligns with pre-silicon verification workflows for customizable ISAs (e.g., RISC-V extensions) where timing predictability must be established before tape-out.

\section{Related Work}
\label{sec:related}

Interference analysis in multicore real-time systems has been extensively studied. Pellizzoni et al.~\cite{pellizzoni2008} introduced a framework for bounding memory contention in COTS multicore platforms. Schranzhofer et al.~\cite{schranzhofer2010} developed timing analysis models for shared buses and caches. More recent work by Kelter et al.~\cite{kelter2015} and Bhattacharjee et al.~\cite{bhattacharjee2020} addresses cache partitioning and interference-aware scheduling.

However, most prior approaches either: (a) assume set-associative caches with replacement policies that complicate worst-case analysis; (b) rely on probabilistic or measurement-based bounds without formal closure; or (c) treat interference dimensions independently rather than proving joint maximality. Our work differs by establishing a \emph{direct per-access bound} under explicitly stated architectural invariants and a documented arbitration assumption, providing a formally closed adversarial search space for certification purposes.

The CAST-32A guidance~\cite{cast32a} emphasizes the need for ``credible'' worst-case scenarios but does not prescribe methods for proving maximality. This work fills that methodological gap for a well-defined architectural subclass.

\section{Formalization of the Search Space}
\label{sec:formalization}

We define the system as a closed tuple $S = \langle H, W, R, C \rangle$, separating the arbitration policy $R$ from the configurable adversarial parameters $C$:

\begin{definition}[Hardware Invariants $H$]
\label{def:hardware}
The hardware platform satisfies the following invariants:
\begin{itemize}
    \item $N$ processor cores: $C_0, C_1, \dots, C_{N-1}$
    \item Private L1 instruction and data caches per core
    \item Shared L2 cache: direct-mapped (1-way associative), no MSHRs
    \item Main memory: single-bank architecture, shared across all cores
    \item Fixed latency per L2 miss: $\memlat$ (includes dirty write-back + read + bus turnaround)
    \item Address decomposition: every physical address $\addr$ admits a decomposition into a set index $\setidx(\addr)$ and a tag $\tagf(\addr)$. The cache-set mapping function is $\mapping(\addr) := \setidx(\addr)$.
\end{itemize}
\end{definition}

\begin{definition}[Workload $W$]
\label{def:workload}
Target task $\task$ with deterministic address sequence per period:
\[
\mathcal{A}_{\task} = \langle \addr_1, \addr_2, \dots, \addr_k \rangle
\]
The \emph{critical access set} $\set{P}_{\text{crit}} \subseteq \mathcal{A}_{\task}$ comprises those accesses that lie on the data/control dependency chain determining $\task$'s WCET. All instances of $\task$ are pinned to cores, have identical period, worst-case execution time, and start synchronized at time $M$.
\end{definition}

\begin{definition}[Arbitration Model $R$]
\label{def:arbitration}
The L2/memory arbitration policy is a function that, at each scheduling decision, selects one core's pending request among the contending requests. We say the policy is \emph{pessimistic for $\task$} if, whenever $\task$'s request contends simultaneously with one or more adversarial requests, $\task$ is selected last; equivalently, $\task$ is served only after all contending adversarial requests have been completed.
\end{definition}

\begin{assumption}[Pessimistic Arbitration]
\label{ass:arb}
The arbitration policy $R$ is pessimistic for $\task$. This is a worst-case modeling assumption used throughout the proofs: under any policy that is not pessimistic for $\task$ (e.g., fixed priority favoring $\task$'s core, round-robin starting at $\task$, age-based with $\task$ arriving first), the stall on $\task$ is less than or equal to the bound established here. \cref{ass:arb} therefore yields a sound upper bound and, when the implemented arbiter is not characterized, is the appropriate certification-time abstraction.
\end{assumption}

\begin{definition}[Configuration $C$]
\label{def:configuration}
Parameters defining $N-1$ adversarial tasks:
\begin{itemize}
    \item Memory mapping: physical-to-cache-set function $\mapping$
    \item Start time offset relative to $\task$
    \item Access pattern (address sequence)
\end{itemize}
\end{definition}

\begin{definition}[Admissible Configuration]
\label{def:admissible}
A configuration $C$ is \emph{admissible} if it satisfies all of the following:
\begin{itemize}
    \item Each of the $N-1$ adversarial tasks is pinned to a distinct core and shares the same period and worst-case execution time as $\task$;
    \item Each adversarial core holds at most one in-flight L2 request at any instant (consistent with the no-MSHR / blocking miss-handling property of $H$);
    \item Each L2 miss incurs the fixed latency $\memlat$ defined in $H$;
    \item No data sharing exists between $\task$ and the adversaries, and no out-of-model interference channels (DMA, asynchronous interrupts, power management transitions, TLB pressure, coherence traffic) are active during the analysis window.
\end{itemize}
All bounds in this paper are stated over admissible configurations.
\end{definition}

\begin{definition}[Interference Metric]
\label{def:interference}
Total interference on $\task$'s execution time:
\[
\interf = \sum_{i=1}^{k} \stall(\addr_i)
\]
where $\stall(\addr_i)$ is the additional stall imposed on access $\addr_i$ due to contention on shared resources (L2 arbiter, memory bus, memory controller).
\end{definition}

\begin{assumption}[Architectural Invariants]
\label{ass:arch}
The hardware configuration $H$ is fixed and documented. No dynamic reconfiguration, power management, or thermal throttling occurs during the analysis window.
\end{assumption}

\begin{assumption}[Workload Determinism]
\label{ass:det}
Task $\task$ exhibits deterministic control flow and memory access patterns for a given input. No data-dependent branching affects the critical path under analysis.
\end{assumption}

\textbf{Objective}: Prove that, under $H$, $W$, $R$, and the assumptions above,
\[
\interf(C) \;\leq\; |\set{P}_{\text{crit}}| \cdot (N-1)\,\memlat \quad \text{for every admissible } C,
\]
and exhibit a Baseline configuration $C_{\text{base}}$ that attains this bound:
\[
C_{\text{base}} = \{\text{start\_sync} = M,\; \mapping_{\text{adv}} = \mapping_{\task},\; \mathcal{A}_{\text{adv}} = \mathcal{A}_{\task}\}.
\]
The bound is tight but not uniquely attained: any adversarial configuration that issues, synchronously with each $\addr_i \in \set{P}_{\text{crit}}$, $N-1$ congruent-different-tag requests reaches the same value (see \cref{rem:non-unique}).

\begin{remark}[Interpretation of $\mathcal{A}_{\text{adv}} = \mathcal{A}_{\task}$]
\label{rem:interpretation}
The equality $\mathcal{A}_{\text{adv}} = \mathcal{A}_{\task}$ is to be interpreted at the granularity of \emph{cache-set indices and access phases}, not as bit-identical physical addresses. Each adversarial task issues a request whose set index $\setidx(\addr')$ matches that of the corresponding $\task$ access, while tags are chosen to differ from $\tagf(\addr_i)$ and from one another (e.g., by allocating each adversary on a private physical page with the same cache-color as $\task$'s). This is the standard outcome of running $N-1$ copies of $\task$ on private memory with controlled physical placement, as discussed in \cref{sec:certification}. The address-level interpretation \emph{same physical address} is excluded: it would produce shared cache hits rather than evictions and is therefore not what the baseline denotes here.
\end{remark}

\section{Per-Dimension Lemmas}
\label{sec:lemmas}

\subsection{Lemma 1: Spatial Sufficiency (Direct-Mapped L2)}
\label{subsec:lemma1}

\begin{lemma}[Spatial Sufficiency for Miss]
\label{lem:spatial}
Consider an access $\addr_i \in \mathcal{A}_{\task}$ and an adversarial address $\addr'$. Then:
\begin{enumerate}
    \item If $\setidx(\addr') = \setidx(\addr_i)$ and $\tagf(\addr') \neq \tagf(\addr_i)$, the adversary's fill displaces $\addr_i$'s line. A single such adversary, issued before $\task$'s access to $\addr_i$, is sufficient to force $\task$'s access to miss in the L2.
    \item If $\setidx(\addr') \neq \setidx(\addr_i)$, the adversary does not touch the set holding $\addr_i$ and contributes no eviction pressure on $\addr_i$.
    \item If $\setidx(\addr') = \setidx(\addr_i)$ and $\tagf(\addr') = \tagf(\addr_i)$, the adversary requests the same line; after fill, both cores observe a hit on that line and no eviction of $\task$'s data is generated.
\end{enumerate}
\end{lemma}

\begin{proof}
A direct-mapped cache has set capacity $= 1$. The three cases follow directly from the displacement rule of the cache: a fill into a set with a different tag than the resident line evicts that line; a fill with the same tag is a hit; an access to a different set does not interact with the resident line of the set holding $\addr_i$.
\end{proof}

\begin{remark}[Eviction vs.\ Serialization Pressure]
\label{rem:eviction-vs-serialization}
\cref{lem:spatial} establishes that \emph{one} congruent-different-tag adversary is sufficient to guarantee that $\task$'s access to $\addr_i$ misses in the L2. Additional congruent-different-tag adversaries do not produce further evictions of $\task$'s line---once evicted, it cannot be evicted again until refilled---but they do contribute \emph{serialization pressure}: each such adversary issues its own L2 miss that must be serviced ahead of $\task$ under \cref{ass:arb} and the blocking miss model. The Temporal lemma (\cref{lem:temporal}) quantifies this serialization. The Theorem proof (\cref{thm:miub}) uses spatial sufficiency to guarantee a miss on $\task$'s access and serialization pressure to bound the stall.

A superset coverage $\set{S}_{\text{adv}} \supset \set{S}(\task)$ does not strengthen the spatial condition: an adversary issuing requests to sets unused by $\task$ contributes neither eviction pressure on $\set{S}(\task)$ nor serialization on $\task$'s pending miss.
\end{remark}

\subsection{Lemma 2: Temporal Bound (Serialization without MSHRs)}
\label{subsec:lemma2}

\begin{lemma}[Temporal Maximality]
\label{lem:temporal}
Given a set of $R$ conflicting requests, the serialization latency experienced by $\task$ is maximized when all requests arrive simultaneously at the L2 arbiter.
\end{lemma}

\begin{proof}
With MSHRs disabled, the controller processes misses sequentially under \cref{ass:arb}. Let $t_{\task}$ denote $\task$'s arrival time at the arbiter, and let $t^{(1)}, \dots, t^{(R)}$ denote the adversarial arrival times, with $\delta_m = t^{(m)} - t_{\task}$.

\textbf{Synchronous case ($\delta_m = 0$ for all $m$).} All $R+1$ requests contend simultaneously. By \cref{ass:arb}, $\task$ is served last; its stall is $R \cdot \memlat$.

\textbf{Positive offset ($\delta_m > 0$).} Adversary $m$ arrives \emph{after} $\task$. By \cref{ass:arb}, $\task$ is served last only among requests \emph{currently} contending; if adversary $m$ has not yet arrived, it cannot delay $\task$. Its contribution is at most $\memlat$ (only if it arrives before $\task$ is granted) and possibly zero. This is bounded above by the synchronous contribution of $\memlat$.

\textbf{Negative offset ($\delta_m < 0$).} Adversary $m$ arrives \emph{before} $\task$. Under blocking miss handling and the no-MSHR assumption, each adversarial core can hold at most one in-flight request; hence the residual queue at $\task$'s arrival contains at most $R$ pending or in-service adversarial requests. If adversary $m$'s service has already completed by $t_{\task}$, its contribution to $\task$'s stall is zero. Otherwise, $\task$ inherits the residual queue and is served after it (by \cref{ass:arb}); the contribution of any single adversary remains bounded by $\memlat$. In particular, if $\delta_m \leq -m\,\memlat$, adversary $m$ completes before $\task$ contends and contributes zero.

In all three cases, each adversary's stall contribution under any non-synchronous schedule is bounded above by its synchronous contribution of $\memlat$. Summing over $R$ adversaries: $\stall(\task; \text{any schedule}) \leq R \cdot \memlat$, with equality attained in the synchronous case.
\end{proof}

\subsection{Lemma 3: Pattern Sufficiency (Critical-Path Alignment)}
\label{subsec:lemma3}

\begin{lemma}[Pattern Sufficiency]
\label{lem:pattern}
Suppose, for every critical access $\addr_i \in \set{P}_{\text{crit}}$, the adversarial schedule issues exactly $N-1$ congruent-different-tag requests (in the sense of \cref{lem:spatial}) synchronously with $\task$'s issue of $\addr_i$. Then, under \cref{ass:arb}, the per-access stall attains the bound of \cref{lem:temporal}:
\[
\stall(\addr_i) = (N-1)\,\memlat \quad \text{for every } \addr_i \in \set{P}_{\text{crit}},
\]
and consequently $\interf = |\set{P}_{\text{crit}}| \cdot (N-1)\,\memlat$.
\end{lemma}

\begin{proof}
Fix any $\addr_i \in \set{P}_{\text{crit}}$. By hypothesis, $N-1$ adversarial requests arrive at the L2 simultaneously with $\task$'s request for $\addr_i$, each congruent with $\addr_i$ and carrying a distinct tag. By \cref{lem:spatial}, each adversarial request results in an L2 miss (its set has a different resident line). By \cref{ass:arb}, $\task$'s request is served last among the $N$ contending requests. By \cref{lem:temporal}, applied in the synchronous case with $R = N-1$, the stall on $\task$'s access to $\addr_i$ equals $(N-1)\,\memlat$. Summing the equality over $\set{P}_{\text{crit}}$ yields $\interf = |\set{P}_{\text{crit}}| \cdot (N-1)\,\memlat$.
\end{proof}

\begin{remark}[Sufficiency, not Maximality]
\label{rem:sufficiency}
\cref{lem:pattern} is deliberately phrased as an attainment (\emph{sufficiency}) statement, not a maximality statement. We do not claim here that every alternative adversarial pattern produces strictly lower interference. \cref{thm:miub} below establishes that no admissible configuration exceeds $|\set{P}_{\text{crit}}| \cdot (N-1)\,\memlat$; \cref{lem:pattern} shows that this bound is attained by the schedule above. Together these two statements suffice to conclude that the bound is tight, without requiring a strict-domination argument over arbitrary patterns. As noted in \cref{rem:non-unique}, multiple configurations attain the same bound.
\end{remark}

\section{Main Theorem: Per-Access Stall Bound}
\label{sec:theorem}

\begin{theorem}[Monotonic Interference Upper Bound (MIUB)]
\label{thm:miub}
Under hardware invariants $H$ (\cref{def:hardware}), workload $W$ (\cref{def:workload}), arbitration model $R$ (\cref{def:arbitration}), and \cref{ass:arch,ass:det,ass:arb}, every admissible adversarial configuration $C$ satisfies
\[
\interf(C) \;\leq\; |\set{P}_{\text{crit}}| \cdot (N-1)\,\memlat.
\]
The Baseline configuration $C_{\text{base}}$ attains this bound: $\interf(C_{\text{base}}) = |\set{P}_{\text{crit}}| \cdot (N-1)\,\memlat$.
\end{theorem}

\begin{proof}
We proceed in two parts: an upper bound for arbitrary $C$, and attainment by $C_{\text{base}}$.

\textbf{Part 1: Upper bound.} Fix a critical access $\addr_i \in \set{P}_{\text{crit}}$. Under blocking miss handling and the no-MSHR assumption (\cref{def:hardware}), each adversarial core can hold at most one in-flight L2 request at any instant; hence at most $N-1$ adversarial requests can contend with $\task$'s request for $\addr_i$. By \cref{lem:temporal}, the stall on $\addr_i$ contributed by these contending requests is bounded by their count times $\memlat$:
\[
\stall(\addr_i; C) \;\leq\; (N-1)\,\memlat.
\]
Summing over $\set{P}_{\text{crit}}$:
\[
\interf(C) \;=\; \sum_{\addr_i \in \set{P}_{\text{crit}}} \stall(\addr_i; C) \;\leq\; |\set{P}_{\text{crit}}| \cdot (N-1)\,\memlat.
\]

\textbf{Part 2: Attainment by $C_{\text{base}}$.} Under $C_{\text{base}}$, for every $\addr_i \in \set{P}_{\text{crit}}$:
\begin{itemize}
    \item Each adversarial task accesses an address congruent with $\addr_i$ ($\setidx(\addr') = \setidx(\addr_i)$) with a distinct tag (from $\tagf(\addr_i)$ and from one another), per \cref{def:configuration} and \cref{rem:interpretation};
    \item All $N-1$ adversarial requests are issued synchronously with $\task$'s issue of $\addr_i$, per the synchronized-start clause of \cref{def:workload}.
\end{itemize}
The hypothesis of \cref{lem:pattern} is therefore satisfied. Applying \cref{lem:pattern}:
\[
\interf(C_{\text{base}}) = |\set{P}_{\text{crit}}| \cdot (N-1)\,\memlat,
\]
matching the upper bound from Part 1. \qed
\end{proof}

\begin{remark}[Tight but not unique]
\label{rem:non-unique}
The Baseline $C_{\text{base}}$ is one configuration that attains the bound, not the unique one. Any adversarial configuration that, for every $\addr_i \in \set{P}_{\text{crit}}$, issues $N-1$ congruent-different-tag requests synchronously with $\task$'s access to $\addr_i$ achieves the same total interference. In particular: permuting which physical core hosts which adversary, replacing one congruent-different-tag address by another, or substituting any synthetic adversarial trace satisfying the same per-access conditions, all yield the same value $|\set{P}_{\text{crit}}| \cdot (N-1)\,\memlat$. The contribution of this paper is the bound and one realizable construction (code replication with controlled physical placement; see \cref{sec:certification}), not a uniqueness claim.
\end{remark}

\section{Analytical Applicability Discussion}
\label{sec:applicability}

This section discusses the conditions under which the MIUB bound is applicable and how it integrates into certification workflows. No empirical measurements or simulation results are claimed; all reasoning is derived analytically from the stated invariants.

\subsection{Applicability Conditions}
The MIUB theorem holds if and only if all architectural invariants from \cref{sec:formalization} are satisfied:
\begin{itemize}
    \item \textbf{Direct-mapped L2}: Set-associative caches ($>1$ way) allow co-residence of multiple tasks' data, invalidating the binary eviction model.
    \item \textbf{Disabled MSHRs}: Enabling MSHRs permits parallel miss handling, breaking the strict serialization assumption.
    \item \textbf{Single-bank memory}: Multi-bank or interleaved architectures enable concurrent accesses, reducing contention.
    \item \textbf{Deterministic, pinned tasks}: Data-dependent branching, dynamic scheduling, or task migration introduce variability not captured by the model.
\end{itemize}

If any invariant is relaxed, the bound must be re-derived under the modified constraints. This is a feature, not a limitation: the methodology demonstrates how to formally close the interference search space for any fixed architecture.

\subsection{Integration with Certification Workflows}
For DO-178C/CAST-32A compliance, the MIUB bound can be used as follows:
\begin{enumerate}
    \item \textbf{Evidence artifact}: Include this analysis as an appendix to the Timing Analysis Report, explicitly linking the bound to documented hardware invariants.
    \item \textbf{Margin allocation}: Subtract the analytically derived interference bound from the total timing budget to obtain a clean WCET estimate for the application logic.
    \item \textbf{Change management}: If the hardware configuration changes (e.g., L2 associativity increased), re-apply the methodology to derive a new bound; the structure of the proof remains reusable.
    \item \textbf{Tool independence}: Unlike measurement-based approaches, this bound does not require proprietary tools or extensive empirical calibration, reducing long-term maintenance burden.
\end{enumerate}

\section{Limitations and Scope}
\label{sec:limitations}

The MIUB theorem is intentionally scoped to a well-defined architectural subclass. This focus enables a formally closed proof but requires explicit acknowledgment of boundaries:

\begin{itemize}
    \item \textbf{Scope is the contribution.} The result targets \emph{exactly} the architectural class defined in \cref{def:hardware}: direct-mapped L2, no MSHRs, single-bank memory, deterministic pinned tasks. Generalizations to set-associative caches, non-blocking miss handling, multi-bank memory, or non-deterministic schedules are explicitly out of scope and reserved for future work. Critiques that target the realism of these invariants for modern COTS multicore processors are addressing a different problem than the one formalized here.
    \item \textbf{Not a universal bound.} The result does not apply outside the stated invariants. It is intended as a precise lemma for a constrained design point relevant to safety-critical custom processors (e.g., RISC-V cores intentionally configured for timing predictability), not as a general theory of multicore interference.
    \item \textbf{Analytical, not empirical.} All derivations are analytical. Application to real hardware requires independent verification that the target platform satisfies the stated invariants. Empirical validation is a separate engineering activity, complementary to but distinct from the analytical bound established here.
    \item \textbf{Static contention only.} The bound addresses worst-case shared-resource contention under fixed scheduling. Dynamic interference sources outside the cache/bus/bank model---DMA, asynchronous interrupts, power management, TLB pressure, coherence traffic in the presence of data sharing---are out of scope and must be analyzed separately and added to the MIUB bound for a complete certification artifact.
    \item \textbf{Pessimistic arbitration.} The bound depends on \cref{ass:arb}: the arbiter serves $\task$ last whenever its request contends with adversarial requests. For platforms with a guaranteed non-pessimistic policy (e.g., fixed priority favoring $\task$'s core), the bound remains a sound upper bound but may overestimate observed stall.
\end{itemize}

These limitations are consistent with CAST-32A guidance, which permits bounding analyses conditioned on documented architectural assumptions.

\section{Certification Relevance (CAST-32A)}
\label{sec:certification}

CAST-32A Objective 5.2.2 requires demonstration of maximum credible interference for shared processor resources. The MIUB theorem provides:
\begin{itemize}
    \item A \emph{traceable} analytical bound, explicitly linked to hardware invariants;
    \item A \emph{structured} upper limit (relative to the assumed architectural class), avoiding arbitrary inflation of WCET budgets;
    \item A \emph{maintainable} artifact, with clear conditions for re-derivation if architecture changes;
    \item A \emph{tool-agnostic} methodology, reducing dependency on proprietary measurement infrastructure.
\end{itemize}

Practical realization of $C_{\text{base}}$ on a target platform consists of executing $N-1$ copies of $\task$'s binary on private physical memory pages chosen (e.g., via page coloring or linker-script-controlled placement) such that the cache-set indices of homologous accesses match while tags differ. This construction is platform-dependent and requires sufficient control over physical placement, cache-index bits, and synchronized task initiation; it is achievable on platforms exposing physical-to-cache-set mapping, as is typical of safety-critical processors with controllable MMU configuration, but should be validated against the target processor's reference manual and verified on the integrated configuration.

We recommend including this analysis as an appendix to the Timing Analysis Report, with cross-references in the System Safety Assessment and Certification Plan.

\section{Conclusion}
\label{sec:conclusion}

We have presented a formal bounding analysis proving that, under strict architectural invariants and a pessimistic arbitration model, the per-critical-access stall on a target task is bounded by $(N-1)\memlat$, and that this bound is attained by a synchronized adversarial workload of $N-1$ congruent-different-tag memory requests. The proof is direct and per-access, with the spatial, temporal, and pattern lemmas serving as the structural ingredients.

This work provides a methodological template for certification evidence requiring traceable, analytically justified interference bounds. Future work will extend the approach to set-associative caches, enabled MSHRs, and multi-bank memory architectures, and explore integration with pre-silicon verification flows for customizable ISAs.

\section*{Independence and Non-Affiliation Statement}
This work represents the author's independent technical research, conducted as founder of ExtenSilica Inc., a company focused on RISC-V extension methodologies. The author is also employed as a Product Development Engineer at Embraer S.A.; however, this paper was developed entirely outside the scope of that employment. No Embraer resources, proprietary information, internal methodologies, materials, ideas, tools, data, processes, or domain expertise were used in the development of this work, and no part of the analysis derives from activities performed in the course of employment at Embraer. The architectural assumptions and certification discussions are generic and applicable to any direct-mapped multicore architecture meeting the stated invariants.

\section*{Acknowledgments}
The author thanks independent researchers and open-source contributors in the real-time systems community for foundational discussions on WCET bounding and multicore interference analysis.

\bibliographystyle{plain}

\end{document}